\documentclass[journal]{IEEEtran}
\usepackage{amssymb,amsfonts,color,fancybox}
\usepackage[mathscr]{eucal}
\usepackage{graphicx}
\usepackage{ifpdf}
\usepackage{longtable}	
\usepackage{multicol}
\usepackage{float}
\usepackage{algpseudocode}
\usepackage{lipsum}
\usepackage[sort,compress]{cite}

\title{Construction of power flow feasibility sets}

\author{Krishnamurthy Dvijotham, Konstantin Turitsyn
\thanks{K.D. is with the California Institute of Technology, e-mail: dvij@caltech.edu. K.T. is with the Massachusetts Institute of Technology, Cambridge, MA, 02139 USA.

}}

\usepackage{verbatim}
\usepackage{algorithm}
\usepackage{bm}
\usepackage{algpseudocode}
\usepackage[all]{xy}
\usepackage{url}
\usepackage{amsmath,amsthm,amssymb}
\usepackage{graphicx}
\usepackage{subcaption}

\newtheorem{theorem}{Theorem}[section]

\newtheorem{lemma}{lemma}
\theoremstyle{definition}
\newtheorem{definition}{Definition}

\theoremstyle{remark}
\newtheorem{remark}{Remark}
\theoremstyle{remark}
\newtheorem{assumption}{Assumption} 

\usepackage{accents}

\newcommand{\nf}{k}

\newcommand{\Poly}[2]{\mathrm{Poly}^{#2}\br{#1}}

\newcommand{\DSF}{DSF}
\newcommand{\SFe}{SF}
\newcommand{\Ly}[1]{\mathcal{L}_y\br{#1}}

\newcommand\munderbar[1]{%
  \underaccent{\bar}{#1}}

\newcommand{\E}{\mathcal{E}}
\newcommand{\G}{\mathrm{pv}}
\newcommand{\NSB}{\mathrm{nsb}}

\newcommand{\Vs}{\mathcal{V}}

\newcommand{\Lo}{\mathrm{pq}}

\newcommand{\Vset}{v}

\newcommand{\Int}{\mathrm{Int}}

\newcommand{\Rep}[1]{\mathrm{Re}\br{#1}}
\newcommand{\Imp}[1]{\mathrm{Im}\br{#1}}

\newcommand{\Vc}{V^c}
\newcommand{\Vml}{\rho}
\newcommand{\Vpl}{\theta}
\newcommand{\Ql}{\underline{q}}
\newcommand{\Qu}{\overline{q}}
\newcommand{\Pl}{\underline{p}}
\newcommand{\Pu}{\overline{p}}

\newcommand{\PQ}{PQ}
\newcommand{\PV}{PV}
\newcommand{\herm}[1]{{#1}^\ast}

\newcommand{\Com}{\mathbb{C}}

\newcommand{\JF}{J_F}

\newcommand{\Hc}{H^{eq}}
\newcommand{\Fc}{H^{op}}

\newcommand{\ub}{\bar{V}}
\newcommand{\lb}{\munderbar{V}}
\newcommand{\fb}{\bar{f}}

\newcommand{\expb}[1]{\exp\left({#1}\right)}

\newcommand{\norm}[1]{\left\lVert {#1} \right\rVert}

\newcommand{\One}{\mathbf{1}}
\newcommand{\tran}[1]{{#1}^T}

\newcommand{\Scal}{\mathcal{S}}

\newcommand{\ic}{\mathbf{j}}

\newcommand{\detb}[1]{\det\left({#1}\right)}
\newcommand{\tranb}[1]{{\left({#1}\right)}^T}
\newcommand{\br}[1]{\left({#1}\right)}

\newcommand{\inv}[1]{{\left(#1\right)}^{-1}}

\newcommand{\R}{\mathbb{R}}

\newcommand{\C}{\mathcal{C}}

\newcommand{\FHinf}[1]{q_\infty}
\newcommand{\FHtwo}[1]{q_2}
\newcommand{\FHone}[1]{q_1}

\begin{document}
\maketitle
\begin{abstract}
	We develop a new approach for construction of convex analytically simple regions where the AC power flow equations are guaranteed to have a feasible solutions. Construction of these regions is based on efficient semidefinite programming techniques accelerated via sparsity exploiting algorithms. Resulting regions have a simple geometric shape in the space of power injections (polytope or ellipsoid) and can be efficiently used for assessment of system security in the presence of uncertainty. Efficiency and tightness of the approach is validated on a number of test networks. 
\end{abstract}
\section{Introduction}
Future power systems relying on large amounts of clean and renewable generation will operate at highly increased levels of uncertainty. Lack of controllability and predictability of renewable generation output will require substantial revision of modern planning and operational practices. On the  operational level, unexpected variations of wind and solar power can potentially compromise the otherwise secure systems. From the long-term perspective, heavy reliance on weather dependent distributed generation sources increases the spatial variability of load levels, and compromises the validity of approaches based on scenario analysis.

In recent years there has been an explosion in the number of academic works that extend the existing procedures by incorporating uncertainty. Most common approaches include stochastic programming, robust and chance-constrained counter-parts of traditional optimal power flow problem \cite{Zhang:2011ci,Sousa:2011jd,Jabr:2013jq,Bent:2013bb,Cao:2013gs}. Despite the progress achieved by the academic community, the proposed algorithms are usually computationally prohibitive and may not be adopted by industry in the near future. Moreover, most of the existing approaches rely on linearized versions of power flow equations and may not be suitable for operational security assessment purposes where nonlinear effects dominate.

Our work attempts to address the need for computationally tractable tools to assess the effect of renewable uncertainty on power system security. The key contribution of our work is a technique for characterization of a maximal set of uncertain power injections that can be tolerated by the power system while maintaining feasibility. We consider two types of uncertainty sets: first is a polytope, where the power injection on a number of buses can vary independently. The second is an ellipsoid which represents chance constraints in situations where variability of individual resources is correlated and has a Gaussian probability distribution. For both of the sets our construction guarantees that the AC power flow equations will have feasible solution for any injections in the uncertainty set. Moreover, by construction the set is a maximal for a given shape. In other words any uniform rescaling of the set will result in violation of one of the constraints or disappearance of the solution. 


Characterization of the feasibility region has been addressed by multiple authors in the recent years. Several studies in the last decades have looked at the question of convexity of the feasibility set \cite{Lesieutre:2005hc,Makarov:2008dj,Zhang:2013je,Lavaei:2014dk}. Non-convexity of the feasibility region for general networks has been explicitly demonstrated in \cite{Lesieutre:2005hc,Makarov:2008dj}. In the last decade, however, much progress has been made in understanding sufficient conditions for convexity of the feasibility region \cite{Zhang:2011ds,Zhang:2013je,Molzahn:2014jpa,Lavaei:2014dk} relevant for global optimality of optimal power flow relaxations. Under the assumption that the network is lossless and all buses are \PV~buses, a number of papers have studied necessary and sufficient conditions that guarantee the existence of power flow solutions \cite{dorfler2014synchronization}. In recent work \cite{BitarLouca}, the authors propose a general framework to construct Linear Matrix Inequality (LMI)-based inner approximations to the feasible set of a class of quadratically constrained quadratic programs (QCQPs). However, the framework does not deal with arbitrary quadratic equality constraints like the AC power flow equations for a meshed network.

Computationally tractable characterization of the power flow feasibility region in terms of ellipsoids in injection space was proposed in \cite{Saric:2008dr}. This is similar to the regions developed in our paper, but relies on the linearized DC power flow approximation. A dual characterization of the feasibility set in terms certificates of insolvability has been developed recently in \cite{Molzahn:2013gv}. More recently, a do-not-exceed limit strategy was proposed for optimal dispatch instructions for intermittent generation \cite{Zhang:2015di, Zhao:2015bl}, that relies on  similar mathematical constructions as developed in this paper  but limits the analysis only to DC power flow models. 

A conceptually close paper that has largely inspired our effort is the recently published study \cite{Bolognani:2015ek} presenting the construction of certificates for existence of AC power flow solutions. This work relies on well-known Banach fixed point theorem to construct ellipsoidal and polytopic regions where the power flow solutions are guaranteed to exist. This has been extended in \cite{Yu:2015wt} to reduce the conservativeness of the constructed certificates. This manuscript extends and generalizes the previous approaches, and proposes a unified methodology for constructing the certificates for existence of feasible solution to power flow equations. In comparison to previous studies it allows for natural incorporation of all the important constraints on voltage levels and power flows, and at the same time provides a computationally tractable methodology based on semidefinite Programming (SDP) that greatly improves the size of the certified region in comparison to previous studies. 

Presentation of our results is organized as follows. In section \ref{sec:modeling} we introduce the notations and key modeling assumptions employed in the paper. Section \ref{sec:LMI} introduces the mathematical background behind our approach. Applications of the approach in practice are described in \ref{sec:applications}. In section \ref{sec:simulations} we validate the approach via analysis of several standard IEEE cases. We conclude in section \ref{sec:conclusions} by assessing the results and discussing possible extensions of the approach.

\section{Modeling} \label{sec:modeling}
In this section we define the mathematical notations used throughout the paper, and discuss the modeling assumptions behind the construction of certified feasible regions.
\subsection{Notations}
We use $\R$ to denote the set of real numbers, $\Com$ the set of complex numbers. $\R^n,\Com^n$ denote the corresponding Euclidean space in $n$ dimensions. Given a set $\C\subset\R^n$, $\Int\br{\C}$ denotes the interior of the set. Given a complex number $x\in\Com$, $\Rep{x}$ denotes its real part and $\Imp{x}$ its imaginary part. $\herm{x}$ denotes its complex conjugate. $|x|$ is its absolute value and $\angle x \in [-\pi,\pi)$ is its phase. $\One$ denotes the vector with all entries equal to $1$.  The Jacobian of a function $F:\R^n\mapsto \R^m$, denoted as $J_F$, is an $n \times m$ matrix whose $i$-th row is the gradient of $F_i$. Given a finite set $S$, $|S|$ denotes the number of elements in the set.

\subsection{AC power flow model}
We represent the transmission network as a graph $\br{\Vs,\E}$ where $\Vs$ is the set of nodes and $\E$ is the set of edges. In power systems terminology, the nodes represent the buses and the edges correspond to power lines. Buses are denoted by indices $i=0,1,\ldots,n$ and lines by ordered pairs of nodes $\br{i,j}$. We pick an arbitrary orientation for each edge, so that for an edge between $i$ and $j$, only one of $\br{i,j}$ and $\br{j,i}$ is in $\E$. 

The transmission network is characterized by its complex admittance matrix $Y \in \Com^{n\times n}$. $Y$ is symmetric but not necessarily Hermitian. Define $G=\Rep{Y},B=\Imp{Y}$.

Let $V_i$ be the voltage phasor, $p_i$ and $q_i$ denote active and reactive injection at the bus $i$ respectively. $V$ is the vector of voltage phasors at all buses. Three types of buses are considered in this work:
\begin{itemize}
\item\underline{\PV~ buses} where active power injection and voltage magnitude are fixed, while voltage phase and reactive power are variables. The set of \PV~ buses is denoted by $\G$. The voltage magnitude set-point at bus $i\in\G$ is denoted by $\Vset_i$.
\item\underline{\PQ~ buses} where active and reactive power injections are fixed, while voltage phase and magnitude are variables. The set of \PQ~ buses is denoted by $\Lo$.
\item\underline{Slack bus}, a reference bus at which the voltage magnitude and phase are fixed, and the active and reactive power injections are free variables. We choose bus $0$ as the slack bus as a convention.
\end{itemize}
We denote the union of \PV~and \PQ~buses as $\NSB=\G\cup\Lo$. Let $\nf=|\NSB|+|\Lo|$. This is the total number of variables to be solved for in the power flow equations.

\begin{definition}[Valid Voltage Phasor Vector]
A vector $V\in \Com^{n+1}$, indexed by $i=0,\ldots,n$, is said to be a \emph{valid voltage phasor vector} if $|V_i|=\Vset_i$ for each $i\in\G$ and $|V_0|=\Vset_0,\angle V_0=0$. Throughout this paper, we will work only with valid voltage phasors. We denote the constraints on valid voltage phasors as $\Hc\br{V}=0$.  
\end{definition}

\begin{definition}[Injection Vector]
A vector $s \in \R^\nf$ is said to be an \emph{injection vector}. The first $n$ coordinates correspond to active power injections at buses $1$ through $n$, and the last $|\Lo|$ components correspond to the reactive injections at the \PQ~buses.
\end{definition}

\begin{definition}[Power Flow Operator]
  Let $V$ be a valid voltage phasor with $V_i=\expb{\Vml_i+\ic\Vpl_i}$. Define the power flow operator $F$ as
\begin{subequations}
\begin{align}
& [F\br{V}]_i = \sum_{j=0}^n B_{ij}\expb{\Vml_i+\Vml_j}\sin\br{\Vpl_i-\Vpl_j} \nonumber \\ 
&+\sum_{j=0}^n G_{ij}\expb{\Vml_i+\Vml_j}\cos\br{\Vpl_i-\Vpl_j},i=1,\ldots,n\label{eq:F2a}\\
&[F\br{V}]_{n+i} =\sum_{j=0}^nG_{ij}\expb{\Vml_i+\Vml_j}\sin\br{\Vpl_i-\Vpl_j}\nonumber \\
&-\sum_{j=0}^n B_{ij}\expb{\Vml_i+\Vml_j}\cos\br{\Vpl_i-\Vpl_j},i=1,\ldots,|\Lo| \label{eq:F2b}
\end{align}\label{eq:F}	
\end{subequations}
\end{definition}
\begin{definition}[Power Flow Equations]
For any valid voltage phasor vector $V$ and any injection vector $s$, the power flow equations are give by $F\br{V}=s$. The variables solved for are $\begin{pmatrix}\Vpl_{\NSB} \\ \Vml_{\Lo} \end{pmatrix}$. We denote by $\JF\br{V}$ the Jacobian of $F$ with respect to $\begin{pmatrix}\Vpl_{\NSB} \\ \Vml_{\Lo} \end{pmatrix}$.
\end{definition}

We now quote a result from our related paper \cite{DjMonotone} that expresses the Jacobian matrix as a quadratic function of the vector of voltage phasors $V$. We will use this result in the certification procedure described in section \ref{sec:cert}.
\begin{lemma}\label{lem:Jacob}
The power flow Jacobian $\JF\br{V}$ can be written as a quadratic matrix function of the voltage phasors:
\begin{align}
\sum_{i \in \Lo}\Delta_i |V_i|^2+\sum_{\br{i,j}\in\E} \Gamma_{ij} \Rep{V_i\herm{V_j}}+\Psi_{ij} \Imp{V_i\herm{V_j}}
\end{align}
where $\Delta_i,\Gamma_{ij},\Psi_{ij}$ are $\nf \times \nf$ matrices that are functions of the network admittance matrix $Y$. 
\end{lemma}
\subsubsection{Operational Constraints on Voltages}
Apart from the power flow equations that describe the conservation laws and behavior of generators and loads we also consider the operational constraints (line limits, bounds on voltage magnitudes etc.). In particular, we require that the solutions of power flow equations satisfy the following constraints:
\begin{subequations}
\begin{align}
\ub_i^2\geq |V_i|^2\geq \lb_i^2,i=1,\ldots,n \label{eq:Oper1a}\\
|V_{k_1}-V_{k_2}|^2\leq \fb_k^2,k=1,\ldots,m \label{eq:Oper1b} \\
\Ql_i \leq \Imp{V_i\herm{\br{YV}}_i} \leq \Qu_i,i\in\G \label{eq:Oper1c}\\
\Ql_0 \leq \Imp{V_0\herm{\br{YV}}_0} \leq \Qu_0 \label{eq:Oper1d}\\
\Pl_0 \leq \Rep{V_0\herm{\br{YV}}_0} \leq \Pu_0 \label{eq:Oper1e}
\end{align}\label{eq:Oper1}
\end{subequations}
The first two constraints \eqref{eq:Oper1a},\eqref{eq:Oper1b} can be interpreted as standard operational constraints imposed on voltage magnitudes and scaled current flows in a power system. Additionally, we also require the constraints arising from the specification of the power flow problem: upper and lower bounds on the reactive power injections at the \PV~buses (generator reactive limits): \eqref{eq:Oper1c}, upper and Lower bounds on the reactive power injections at the slack bus \eqref{eq:Oper1d} and upper and lower bounds on the active power injection at the slack bus \eqref{eq:Oper1e}. These constraints are collectively denoted as $\Fc\br{V}\geq 0$

Although only these constraints are enforced in our simulations, the technique can be naturally extended to incorporate any constraints defined via quadratic inequalities in $V$.

\section{Algorithmic construction of feasible region} \label{sec:LMI}

Our goal in this paper is to find regions in the space of injections such that every point in the region has a power flow solution satisfying the operational constraints \eqref{eq:Oper1}. We first define this precisely:
\begin{definition}[Domain of Strict Feasibility (\DSF)]
$s$ is said to be \emph{strictly feasible (\SFe)} if there is a solution of the power flow equations $F\br{V}=s$ such that $\Fc\br{V}>0$. We say that $\Scal \subset \R^\nf$ is \emph{domain of strict feasibility (\DSF)} if every $s \in \Scal$ is strictly feasible.
\end{definition}
\begin{remark}
We require strict feasibility to deal with mathematical issues that arise in our certification procedure. From a practical perspective, this does not matter much, since this simply corresponds to make the bounds (line limits, voltage magnitude limits etc.) tighter by a very small amount.
\end{remark}
In this paper, we will assume that $\Scal$ is an ellipsoidal or box region centered at around a nominal strictly feasible  $s^0$. These may be the $0$ injection vector or a nominal operating point known to the system operator to be strictly feasible. We will attempt to find the largest ball around $s^0$ that is a \DSF. The construction is algorithmic and based on the following idea: Let $s\in \Scal$ be arbitrary. We consider the straight line path from $s^0$ to $s$. Let $V^0$ be a solution of $F\br{V}=s$ such that $\Fc\br{V^0}>0$. If $\JF\br{V^0}$ is non-singular, the power flow operator is locally invertible (by the inverse function theorem) in a small ball around $s^0$. Thus, we can take a small (but finite) step from $s^0$ towards $s$, on the line segment connecting them. Let us say the new point is $s^1$. Suppose that $s^1$ is strictly feasible as well, and $s^1=F\br{V^1},\Fc\br{V^1}>0$. If $\JF\br{V^1}$ is non-singular, we can again move a small step towards $s$ to get $s^2$. This procedure can continue as long as we never hit a singularity of the Jacobian $\JF\br{V}$ or a point such that $\Fc\br{V}\not>0$. The following sections illustrate how we can check this condition using semidefinite programming.
The procedure is essentially like the continuation power flow algorithm \cite{ajjarapu1992continuation}. However, our goal is not to find a power flow solution (like in the continuation power flow), but to certify that a strictly feasible power flow solution exists for every $s\in\Scal$. This is illustrated pictorially in figure \ref{fig:Mapping}.
\begin{figure*}
\centering
\includegraphics[height=.2\textheight]{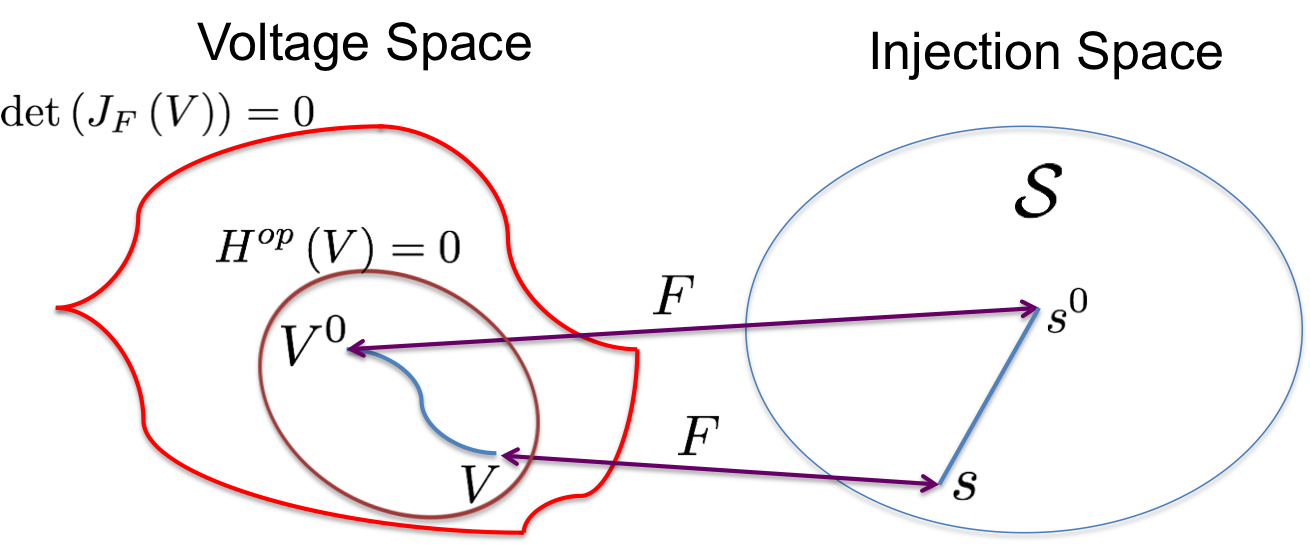}
\caption{Pictorial depiction of our Framework: The red curve is a closed surface in the voltage space defined by $\detb{\JF\br{V}}=0$ and the inner brown curve represents the boundary of the operational constraints on voltages. Inside this region, voltages satisfy $\Fc\br{V}>0$. The blue ellipse denotes the region $\Scal$ in injection space. As we move from $s^0$ to $s$ in injection space, the solution moves from $V^0$ to $V$ in voltage space along the curve. Our certification procedure consists of ensuring that for any $s \in \Scal$, the curve in the voltage space never intersects the red curve or the boundary of the operational constraints $\Fc\br{V}= 0$.}
\label{fig:Mapping}
\end{figure*}
\subsection{Mathematical Characterization of feasible region} \label{sec:cert}
We now formalize the arguments outlined in the previous section. We start by making an assumption on $\Scal$:
\begin{assumption}\label{assum1}
$\Scal$ is convex and $\exists s^0\in\Scal$ that is strictly feasible.
\end{assumption}
\begin{remark}
In most cases, $s^0=0$ satisfies these conditions so this assumption is not problematic. More generally, we will attempt to certify that all injection vectors in a ball around a known strictly feasible nominal injection vector are strictly feasible.
\end{remark}
\begin{theorem}\label{thm:PFCond}
Suppose the following implications hold:
\begin{subequations}
\begin{align}
&\Fc\br{V}\geq 0,\Hc\br{V}=0 \implies \detb{\JF\br{V}}\neq 0\label{eq:MainImplication1}\\
&\Fc\br{V}\geq 0, \Hc\br{V}=0, F\br{V}\in\Scal \implies \Fc\br{V}> 0\label{eq:MainImplication2}
\end{align}\label{eq:MainImplication}
\end{subequations}
 Then, $\Scal$ is a \DSF.
\end{theorem}
\begin{proof}
See Appendix. 
\end{proof}
\subsection{Checking the Conditions Using semidefinite programming}
Written in terms of the vector of voltage phasors, these constraints can be expressed as polynomial equalities or inequalities in $V$. These constraints are generally non-convex constraints and may not be amenable to tractable computation. Thus, we attack the problem using the moment relaxation approach \cite{lasserre2009moments}, leading to a semidefinite programming based sufficient condition for the implication \eqref{eq:MainImplication}. We do this in two steps: We first find a tightening of the operational constraints such that every point in the feasible set defined by the operational constraint also has a non-zero power flow Jacobian - this certifies \eqref{eq:MainImplication1}. Then, we show that if $s\in \Scal$  is feasible injection vector,  then it \emph{must be} strictly feasible - this certifies \eqref{eq:MainImplication2}. 
\subsubsection{Non-Singularity of the Jacobian}
\begin{figure}
\centering
\includegraphics[width=.7\columnwidth]{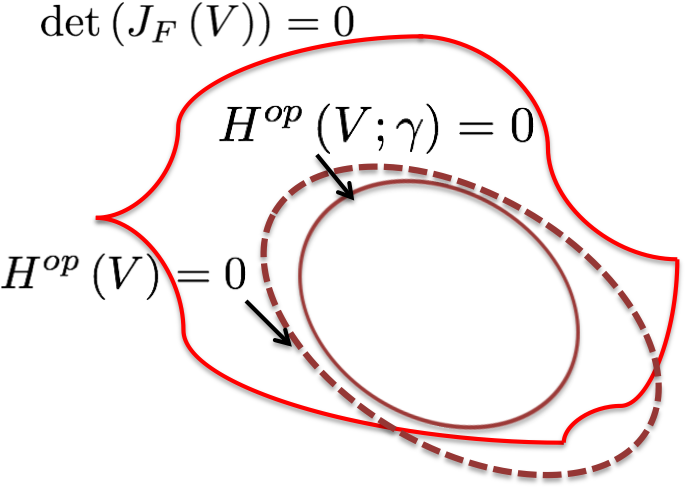}
\caption{Tightening the operational constraints to ensure that $\Fc\br{V;\gamma}\geq 0 \implies \detb{\JF\br{V}}\neq 0$.}
\label{fig:Shrinkage}
\end{figure}
Our first step is to ensure that every valid voltage phasor vector $V$ satisfying the operational constraints $\Fc\br{V}\geq 0$ (and the validity constraint $\Hc\br{V}=0$) also satisfies $\detb{\JF\br{V}}\neq 0$. However, this may not be true. Thus, we also allow for a tightening of the operational constraints such that this becomes true (see figure \ref{fig:Shrinkage}). Specifically, we consider tightening the line limits: We replace $|V_i-V_j|\leq \fb_{ij}$ by $|V_i-V_j|\leq \min\br{\gamma,\fb_{ij}}$ for some $\gamma>0$. Denote the modified operational constraints as $\Fc\br{V;\gamma}\geq 0$. We then solve the following problem:
\begin{align}
&\text{Maximize } \gamma \text{ such that } \nonumber \\
&\Fc\br{V;\gamma}\geq 0, \Hc\br{V}=0 \implies \detb{\JF\br{V}}\neq 0\label{eq:TightenCond}
\end{align}
The implication in \eqref{eq:TightenCond} is equivalent to infeasibility of the following system:
\begin{subequations}
\begin{align}
& \text{Find } \br{z,V}\text{ such that }  \\
& \JF\br{V}z = 0 \\
& \tran{z}z = 1 \\
& \Fc\br{V;\gamma}\geq 0,\Hc\br{V}=0
\end{align}\label{eq:JSing}
\end{subequations}
If this system is infeasible, we know that for each $V$ such that $\Fc\br{V;\gamma}\geq 0,\Hc\br{V}=0$, there is no $z$ such that $\tran{z}z=1$ and $\JF\br{V}z=0$, which means that $\JF\br{V}$ is non-singular. This feasibility problem involves non-convex constraints and is difficult to solve in general. However, all the constraints are defined by polynomial equations and inequalities in $\br{V,z}$. We perform a moment relaxation of this problem \cite{lasserre2009moments}, which replaces the non-convex constraints by weaker convex constraints. If there is no $\br{z,V}$ that satisfy the weaker constraints, then we know that \eqref{eq:JSing} is also infeasible. Checking the infeasibility of the moment relaxation to solving a semidefinite program, for which there are well-known efficient algorithms and software. This conversion is discussed in Appendix section \ref{sec:MomentRelax}. We note the sizes of the resulting SDPs for some IEEE benchmarks in section \ref{sec:simulations}.
\subsubsection{Strict Feasibility of Constraints}
For each constraints $i$ in $\Fc$, we solve the following feasibility problem:
\begin{subequations}
\begin{align}
&\text{ Find } V \text{ such that}\\
& \Fc_i\br{V;\gamma}=0\\
& \Fc\br{V;\gamma}\geq 0,\Hc\br{V}=0,F\br{V}=s,s \in \Scal
\end{align}\label{eq:Fineq}
\end{subequations}
If the above problem is infeasible, then we have a certificate of \eqref{eq:MainImplication2}. The constraints are quadratic functions of $V$. Again, this is a non-convex optimization problem, but it can be relaxed using the moment relaxation approach, described in section \ref{sec:MomentRelax}. The overall certification procedure is summarized in Algorithm \ref{alg:Certification}.

\subsubsection{Tightness of the Certificate}
In practical cases (section \ref{sec:simulations}), we find that the value of $\gamma$ obtained by solving \eqref{eq:TightenCond} often exceeds the practical bounds by a significant margin. Thus, for all practical purposes, we can assume that the tightening in step 1 of the certification procedure does not significantly impact the conservatism in the certificate.
Now let us consider the second part, which involves checking infeasibility of \eqref{eq:Fineq} for each $i$. Suppose that this fails for some $i$. Do we know that $\Scal$ is not a domain of strict feasibility? In other words, can one be assured that there exists $s\in\Scal$ such that no solution of $F\br{V}=s$ satisfies $\Fc\br{V;\gamma}\not> 0$? In general, the answer is no (if it were not, then would be solving an NP-hard problem exactly). However, when the infeasibility test fails, the solver does output $\br{V,s}$ which are feasible for a relaxed version of the constraints in \eqref{eq:Fineq} (see section \ref{sec:MomentRelax}). One can then check if $\br{V,s}$ satisfy the constraints in the problem \eqref{eq:Fineq}. If this is true, one is certain that the set $\Scal$ is not a domain of strict feasibility. In the numerical section \ref{sec:simulations}, we find that this is indeed the case for the test cases we experiment with.
\begin{algorithm}
\caption{Certifcation Procedure}\label{alg:Certification}.
\begin{algorithmic}
\State Find maximum $\gamma$ such that \eqref{eq:JSing} is infeasible.
\State Check that the relaxation of \eqref{eq:Fineq} is infeasible for each $i$.
\end{algorithmic}
\end{algorithm}
\section{Practical applications} \label{sec:applications}
Despite the promising opportunities for acceleration of the SDP solvers described in the previous sections, the resulting problem is still going to be too computationally intensive for real-time applications. Hence, we propose to use the algorithm in two stage procedure. On the first, offline stage one or multiple certificates are constructed for different base operating points $V^0, s^0$ and possibly different topologies of the system. Each of the certificates establishes feasibility of the whole region of injection space, and can be reused multiple times in online applications. 

During the online stage the simple algebraic form of the region definition is used for fast assessment of feasibility of possible operating points and for fast convex optimization of control actions. We envision several applications where the certificates could be used in real-time operation:\\
\underline{Security assessment}. Many of the real-life contingencies can be represented as rapid changes in power injections. These mainly result from triggering of protection systems on loads and distributed generation, although in the future events like unexpected wind gusts and clouds covering major solar plants may have similar effect on the system. Certificates established in this work can be used to screen the safe scenarios from the contingency list, and identify critical contingencies. More generally they naturally establish the levels of uncertainty in power injection vector that can be tolerated by the power system. Whenever the prior distribution of the renewable fluctuations is known, the certificates can be also used to verify the low probabilities of system failure. This can be accomplished by inscription of the injection probability level set in the certified feasibility region. \\
\underline{Establishment of do-not-exceed limits}. As proposed in \cite{Zhao:2015bl} the do-not-exceed (DNE) limit strategy is viable option for limiting the effect of uncontrollable renewables on the security of the system. The certificates can be used for regular updates of DNE limits. Although this strategy can be more conservative in comparison to optimization approaches developed in \cite{Zhang:2015di}, its advantages include reliance on AC power flow equations and effectively zero computational cost. The latter advantages opens the possibility of using the approaches in micro-grid setting with limited computational resources for dispatch decision making.

\section{Simulations}\label{sec:simulations} 

As noted previously, the first part of the certification \eqref{eq:JSing} only depends on the network structure and not on the particular operating point. Thus, as long as the network topology and parameters are fixed, this can be computed once (offline). In table \ref{tab:JacNS}, we have the maximum value of $\gamma$ (a uniform upper bound on $|V_i-V_j|$ for all transmission lines) for three different test networks.  The values obtained show that the bounds are non-conservative and looser than the actual flow constraints imposed in practice. In the subsequent sections, we apply our methodology to three networks included with the Matpower package \cite{zimmerman2011matpower}. We construct polytopic and ellipsoidal regions centered around a nominal solution, and check whether the inscribed regions are tight (in the sense that they cannot be expanded without violating some of the operational constraints). 

\subsection{3-bus system}
We start our discussion with a three-bus system model depicted in figure \ref{fig:3busnetwork}. The example is taken from \cite{MolzahnLimits} and was used to prove that the SDP relaxation of the Optimal Power Flow (OPF) \cite{lavaei2012zero} problem may not be exact . 

All the buses are \PV~buses, so the only controllable injections are the active injections at buses $2$ and $3$ (bus $1$ is the slack bus). The true feasible set (obtained via brute-force discretization) is plotted (in the $p_2$-$p_3$ space, in p.u system with a 100MVA base) in blue in figures \ref{fig:Case3Rect},\ref{fig:Case3Ellipse}. Superposed on it, we plot various subsets $\Scal$ that we certify using the procedure from algorithm \ref{alg:Certification}. In figure \ref{fig:Case3Rect}, we look at $\Scal$ defined in terms of individual bounds on $p_2,p_3$ and plot the largest rectangles with a given center and aspect ratio that can be embedded in the feasible set. The results show that the embedding is tight, ie, one of the corners of the rectangle is typically close to the boundary of the feasible set so that the rectangle cannot be expanded while still being a \DSF.  
 
 In figure \ref{fig:Case3Ellipse}, we plot an ellipse inscribed in the feasible set. The ellipse covers a much larger area relative to the rectangles, since the axes of the ellipse are aligned roughly with the directions in which the feasible set extends. Again, we see that the ellipsoid is non-conservative in the sense that on increasing its radius (while maintaining the aspect ratio), it would no longer be contained within the feasible set.
 
 The ellipsoid embedding also shows how this may be advantageous compared to traditional SDP relaxations of the OPF problem \cite{lavaei2012zero}\cite{low2014convex}. 
The standard SDP relaxation of the OPF problem would work with the convex hull of the blue feasible set, which is denoted by the black curve in figure \ref{fig:Case3Ellipse}. In the top left corner of the blue region, the convex hull contains points not in the original feasible set. If $p_2$ is fixed to $-.6$, and the OPF objective is to minimize $p_3$, it is easy to see that the SDP relaxation of the OPF problem will find a solution that is infeasible for the original problem. On the other hand, with our approach, by first inscribing an ellipse and minimizing $p_3$ subject to the constraint that $\br{p_2,p_3}$ lie in the ellipse, we will obtain a point on the boundary of the ellipse that is feasible and near optimal (since the boundary of the ellipse in that direction is close to the boundary of the feasible set).

\begin{figure}[htb]
        \centering
        \begin{subfigure}[b]{0.5\columnwidth}
                \includegraphics[width=\textwidth]{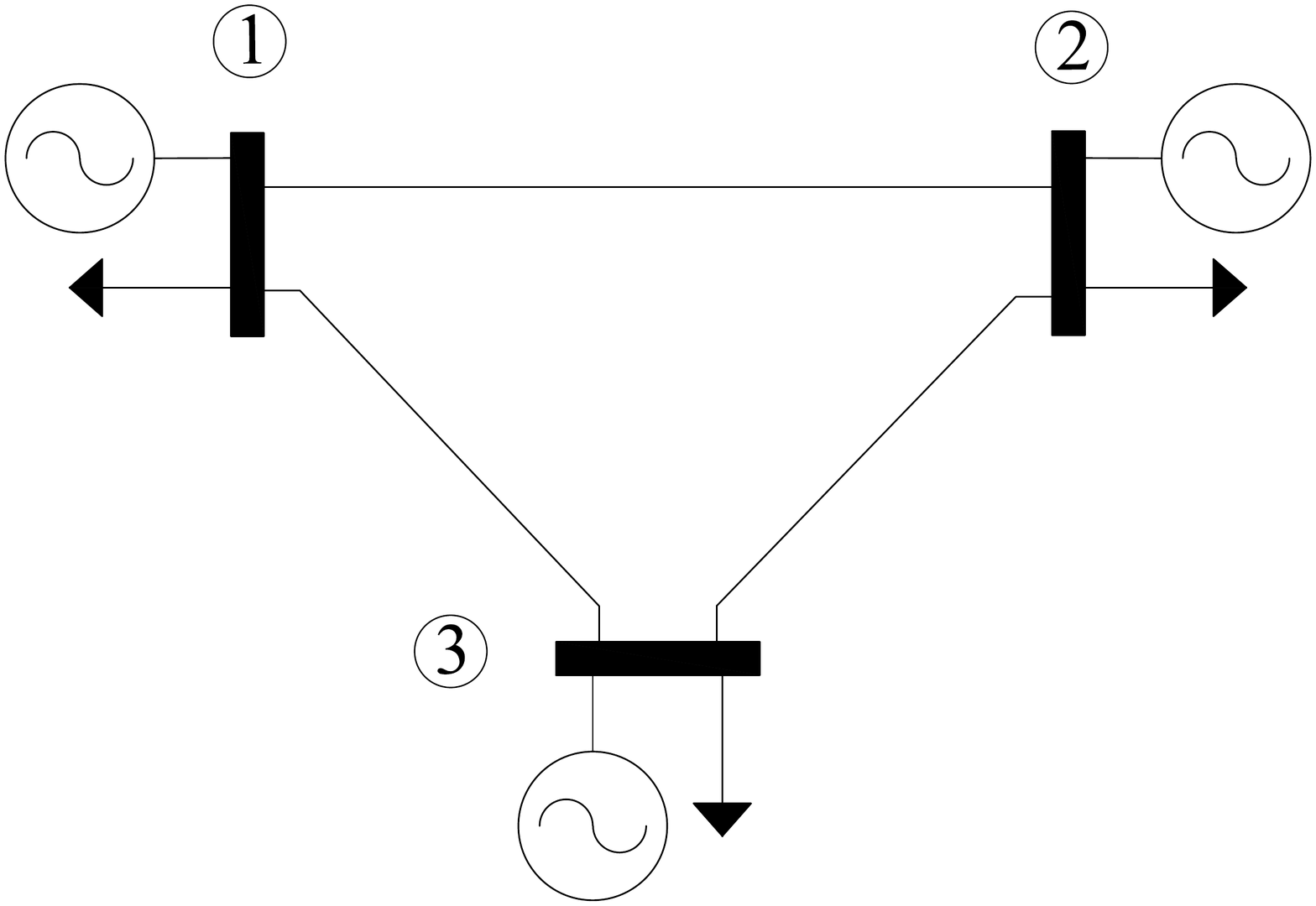}
                \caption{3-bus network}\label{fig:3busnetwork}
        \end{subfigure}
        ~ 
        \begin{subfigure}[b]{0.7\columnwidth}
                \includegraphics[width=\textwidth]{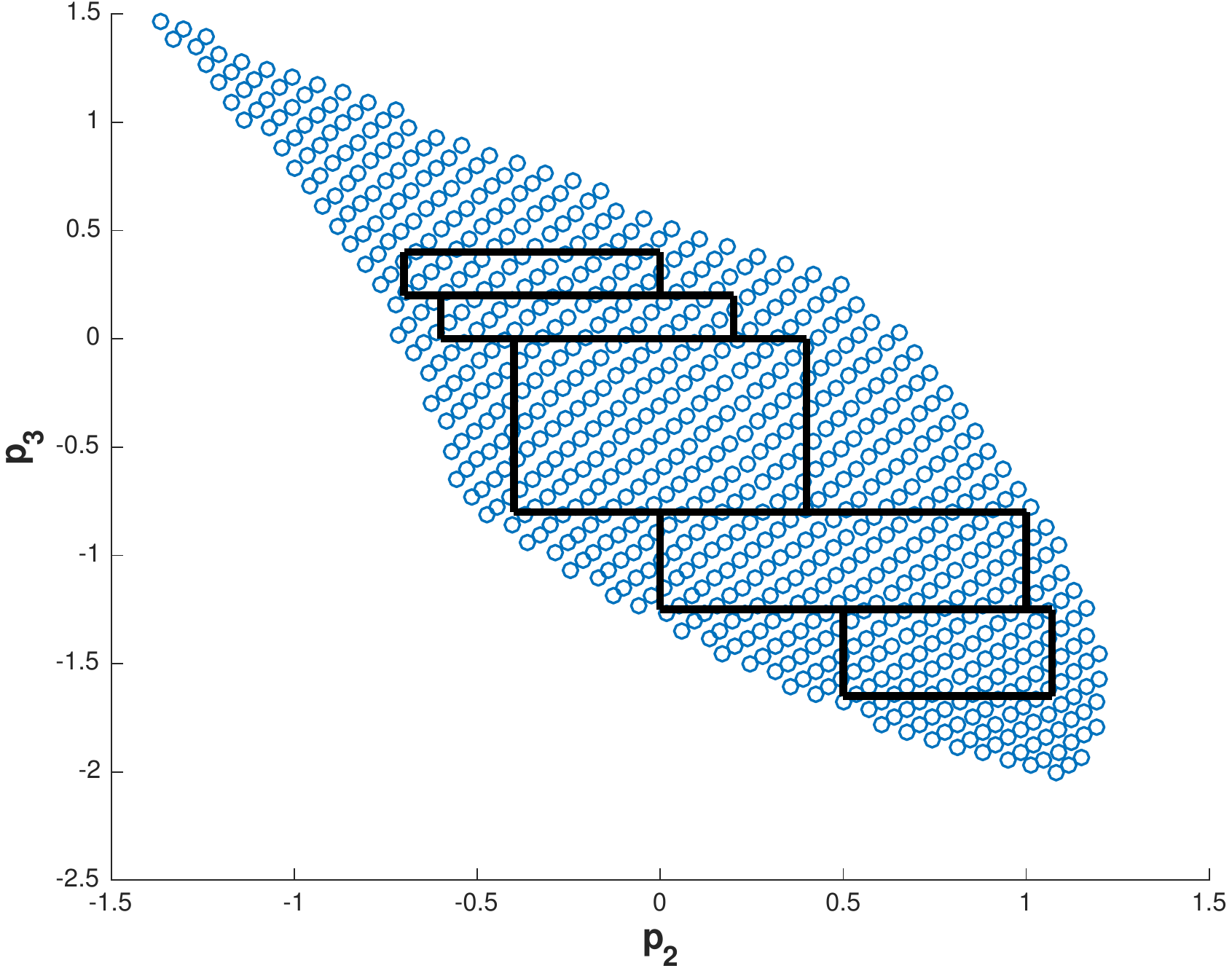}
                \caption{Embedding rectangles in feasible set: Blue region is the true feasible region and each black box is an embedded rectangle in the feasible region.}
                \label{fig:Case3Rect}
        \end{subfigure}
        ~ 
        \begin{subfigure}[b]{0.7\columnwidth}
                \includegraphics[width=\textwidth]{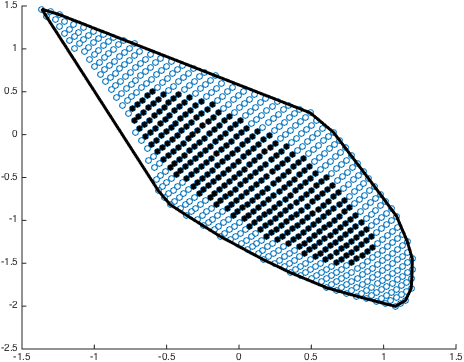}
                \caption{Embedding ellipses in feasible set: Blue region is feasible region and black ellipse is embedded in the feasible region. The black curve denotes the convex hull of the feasible set.}
                \label{fig:Case3Ellipse}
        \end{subfigure}
        \caption{3 Bus Network}\label{fig:animals}
\end{figure}

\subsection{6-bus system}
Next, we consider the 6 bus system shown in figure \ref{fig:6busnetwork}. Buses 1-3 are \PV~buses and buses 4-6 are \PQ~buses. We choose a uniform bound on the voltage differences between neighbors $|V_i-V_j|\leq .4$. \eqref{eq:JSing} is indeed infeasible with this choice of $\gamma$, so that every voltage vector satisfying these constraint has a non-singular Jacobian. The injection parameters are the active injections at buses 2 and 3, and the active and reactive injections at buses 4,5 and 6, leading to a 8-dimensional parameter space. We attempt to embed an a hypercube ($\Scal=\{s:\norm{s}_\infty\leq \delta\}$) inscribed in the feasible set. This is the analog of the rectangular region in the previous case. We are not limited to this set, but we chose this for simplicity.

We compute the maximum $\delta$ such that $\Scal$ remains contained in the feasible region- it turns out to be $.728$ p.u. This means that if all active (at the non-slack buses) and reactive injections (at the \PQ~buses) are bounded by $.729$ (in absolute value), we are guaranteed that there is a PF solution satisfying the operational constraints. On increasing $\delta$ further (to $.73$), we find that the upper bound on $|V_1-V_5|$ is violated (that is \eqref{eq:Fineq} fails to remain infeasible for this constraint). 

\begin{table}[htb]
\begin{center}
\begin{tabular}{|c|c|c|c|}
\hline
System & Case 3 & Case 6 & Case 14 \\
\hline
Maximum $\gamma$ & 1 & .7 & .58 \\
 \hline
\end{tabular}
\end{center}
\caption{Certifying Jacobian Non-Singularity. Maximum Value of $\gamma$ for different networks}
\label{tab:JacNS}
\end{table}

\begin{figure}
        \centering
                \includegraphics[width=.5\columnwidth]{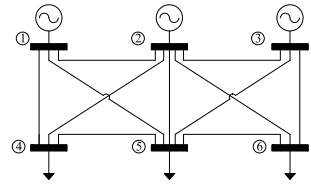}
                \caption{6-bus network}\label{fig:6busnetwork}
        \end{figure}

\subsection{14-bus system}
\begin{figure}[ht]
\centering
    \includegraphics[width=.7\columnwidth]{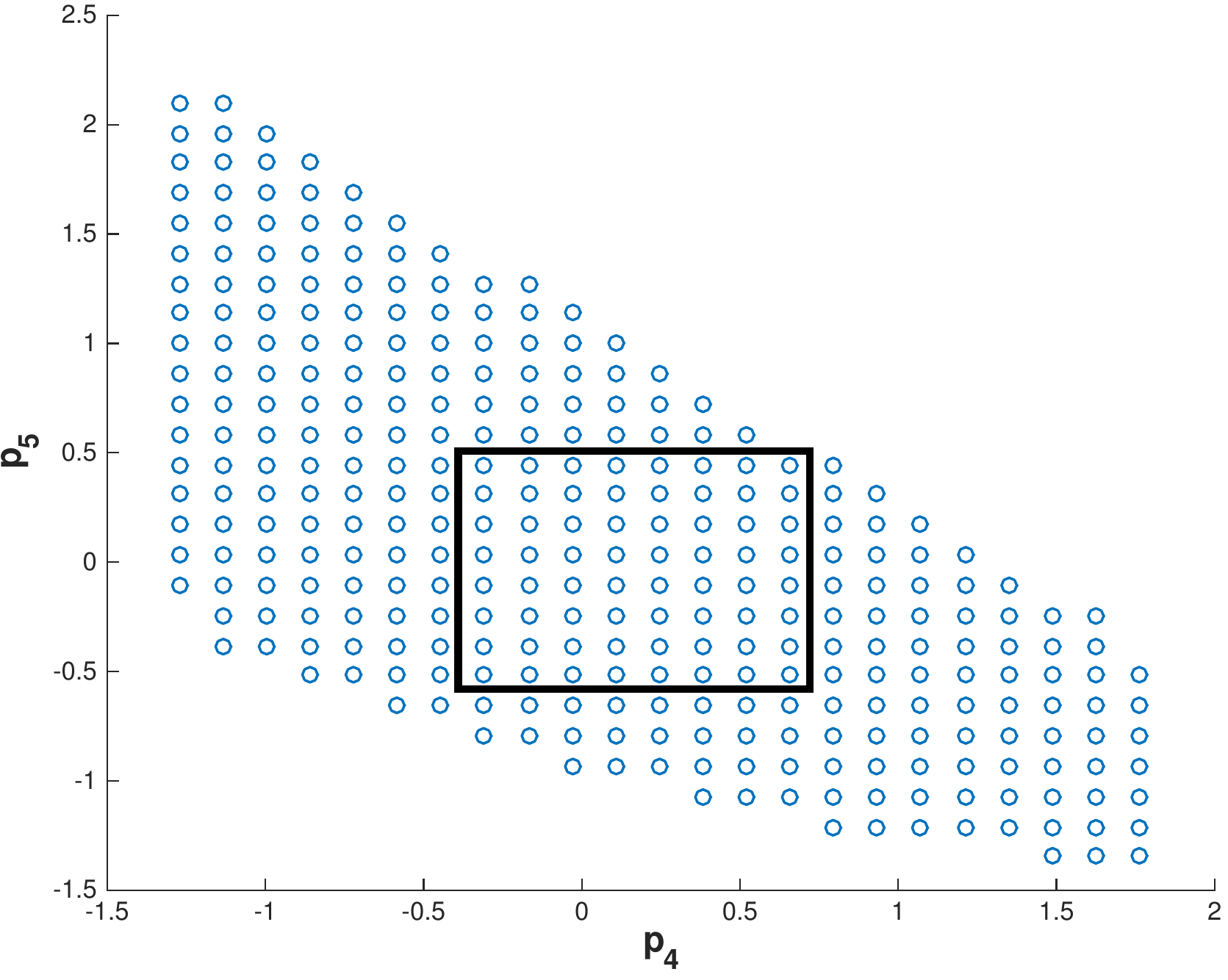}
    \caption{Cross-section of the true feasibility region and the inscribed certificate for IEEE $14$-bus system.}\label{fig:14busnetwork}
\end{figure}
We repeat the same experiment with the IEEE 14 bus network. In this network, we find that a uniform bound on all injections, is too conservative. We thus use a heuristic to pick the relative scaling of the bounds on the injections at each bus. We create random voltage profiles satisfying the operational constraints, and compute the corresponding injection vectors. We then ``fit'' a box region to the resulting samples, and attempt to certify that this region (or a scaled version) is completely contained within the feasible set. 

\begin{table*}
\centering
\begin{tabular}{|c|c|c|c|c|c|c|c|c|c|c|c|c|c|}
\hline
 Bus &  2 & 3 & 4 & 5 & 6 & 7 & 8 & 9 & 10 & 11 & 12 & 13 & 14\\
 \hline
 $p_{\min}$ & -0.2942 &  -0.1455   &-0.3904  & -0.5806  & -0.1618  & -0.2446 &  -0.1108 &  -0.3683  & -0.1274  & -0.0939 &  -0.0446 &  -0.0953 &  -0.0504  \\
 \hline
 $p_{\max}$ & 0.5280 &   0.1752 &   0.7206 &    0.5078 &    0.3206 &    0.3401 &    0.1073 &    0.3285 &    0.3497 &    0.1750 &    0.1292 &    0.2338 &    0.1198 \\
 \hline
\end{tabular}
\centering
\vskip 2em
\begin{tabular}{|c|c|c|c|c|c|c|c|c|c|c|c|c|c|}
\hline
 Bus &   4   &  5  &   7 &     9 &    10 &    11 &    12 &    13 &    14 \\
 \hline
 $q_{\min}$ & 0.2802 &  0.4454 &    0.0253 &   -0.0964 &    0.0210 &    0.0181 &   -0.0086 &    0.0020  &   0.0066\\
 \hline
 $q_{\max}$ & 0.7384  &  0.8873 &   0.1944 & 0.1750 &    0.2547 &    0.1586 &    0.1278 &    0.2213   & 0.1034 \\
 \hline
\end{tabular}
\caption{Bounds on active and reactive power injections at each bus in the IEEE 14 bus network, guaranteeing existence of a feasible power flow solution (in per unit system with a 100 MVA base)}\label{tab:Bounds}
\end{table*}

We plot a projection of the feasible injection onto the $p_{4}-p_{5}$ space (the injections at buses $4$ and $5$) in figure \ref{fig:14busnetwork}. The figure shows that $\Scal$ indeed does touch the boundary of the feasible set and cannot be expanded further. The actual bounds established for each bus in the network are quire reasonable and are shown in table \ref{tab:Bounds}.
\subsection{Computation Time}
Table \ref{lab:CompNS} shows the size of the largest positive semidefinite (PSD) constraint produced while solving the sparsity-exploiting moment relaxations of \eqref{eq:Fineq},\eqref{eq:JSing}. This is usually the dominating factor in the computational effort of solving the moment relaxation. In \cite{molzahn2014sparsity}, the authors show how a similar moment relaxation can be scaled upto a 300 bus network. In newer recent work, the authors have managed to scale this up to several thousand buses\cite{molzahnNew}. 

In terms of actual implementation, we use the convex optimization parser-solver CVX \cite{cvx}\cite{gbas08} and MOSEK \cite{mosek} as the underlying solver. For the $14$-bus system, on a 2014 Macbook Pro with a 2.6GHz Intel Core i7 processor, it takes 20 seconds to solve \eqref{eq:JSing} and each instance of \eqref{eq:Fineq}. Using the ideas from \cite{molzahn2014sparsity,molzahnNew}, we believe that the approach can be scaled to several thousand buses.

\begin{table}[htb]
\begin{center}
\begin{tabular}{|c|c|c|c|c|}
\hline
Num of buses & 3 & 6 & 14 & 30 \\
\hline
Size of PSD & 10 & 55 & 105 & 190  \\
 \hline
\end{tabular}
\end{center}
\caption{Certifying Jacobian Non-Singularity. Maximum Value of $\gamma$ for different networks}
\label{lab:CompNS}
\end{table}
\section{Discussion and Conclusions}\label{sec:conclusions}
We have developed a novel computational approach for constructing geometrically simple regions of existence of feasible solution in the power injection space. Construction of these regions is based on semidefinite programming techniques which were recently proven to be both effective and scalable \cite{lavaei2012zero,molzahn2014sparsity}.   Our numerical experiments on small test networks have indicated that the constructed regions are tight, so any rescaling of a particular shape results in violation of constraints at least for some operating points.
Regions that are certified this way have a very simple geometric shape of a polytope or an ellipse and can be naturally used in a number of security assessment or emergency control applications that require fast decision making. Two particularly suitable applications are the assessment of overload and voltage collapse risks in the presence of renewable generation uncertainty and identification of optimal load shedding actions in emergency situations. 

There are several ways how the approach can be improved and extended. First, the current version of the algorithm does not optimize with respect to the shape of the certified region. Identifying the largest ellipsoid that can be inscribed in a feasibility set may significantly improve the conservativeness of the approach. These techniques can be naturally extended to small-signal stability (analyzing Jacobians of the dynamical equations instead of the power flow equations). Integration of small-signal and transient stability for credible contingencies would complete the full characterization of the safe operation region and will become a powerful tool for system operators.
\section*{Acknowledgements}
We thank Dan Molzahn for insightful discussions and for sharing code on sparse moment relaxations. We also thank Misha Chertkov, Steven Low, Janusz Bialek, Xiaozhe Wang and Hung Nguyen for useful feedback on this manuscript. 

\bibliographystyle{IEEEtran}
\bibliography{feasibilitylibrary}

\section{Appendix}
\subsection{Proof of theorem \ref{thm:PFCond}}
Let $s\in\Scal$ be arbitrarily chosen. Using assumption \ref{assum1}, we know that $\exists s^0\in\Scal$ and a valid voltage phasor vector $V^0$ such that $s^0=F\br{V^0}$, $\Fc\br{V^0;\gamma}>0$. We study the dynamical system \[\frac{d}{dt} {\begin{pmatrix}\br{\angle V}_{\NSB} \\ \br{|V|}_{\Lo}\end{pmatrix}}=\inv{\JF\br{V}}\br{s-F\br{V}}\]
with initial condition $V\br{0}=V^0$. Then $\frac{d}{dt} \br{F\br{V}-s}=-\br{F\br{V}-s}$ so that 
\begin{align}
&F\br{V\br{t}}-s=\expb{-t}\br{F\br{V^0}-s}\implies \nonumber\\
&F\br{V\br{t}}=s(1-\expb{-t})+\expb{-t}s^0 \in \Scal\label{eq:Pf1cond}
\end{align}
By \eqref{eq:MainImplication1}, we know that $\JF\br{V}$ is non-singular for every feasible for every $V$ such that $\Fc\br{V;\gamma}\geq 0,\Hc\br{V}=0$, so the dynamical system is well defined as long as $\Fc\br{V;\gamma}\geq 0$ (the constraint $\Hc\br{V}=0$ is always maintained since the dynamics only changes $\br{\angle V}_{\NSB},\br{|V|}_{\Lo}$). We start with $\Fc\br{V^0}>0$. Suppose that $t^*<\infty$ is the first time instant the dynamical system hits the operational boundary, ie, $\Fc_i\br{V\br{t^*}}=0$ for some $i$. The dynamical system is well-defined upto this point. By implication \eqref{eq:MainImplication2}, we know that since $F\br{V\br{t^*}}\in\Scal,\Fc\br{V\br{t^*}},\Hc\br{V\br{t^*}}=0$, we must have $\Fc\br{V\br{t^*}}>0$, which is a contradiction to $\Fc_i\br{V\br{t^*}}=0$. Thus, the dynamical system always satisfies $\Fc\br{V\br{t}}>0$, so that \eqref{eq:Pf1cond} implies that it must converge to $V$ such that $F\br{V}=s,\Fc\br{V}>0$, so that $s$ is strictly feasible. Since $s\in\Scal$ was chosen arbitrarily, $\Scal$ is a domain of strict feasibility.
\subsection{Moment Relaxation}\label{sec:MomentRelax}
In this section, we describe the moment relaxation of \eqref{eq:JSing}. For any $x\in \R^n$, let $\Poly{x}{i}$ denote the vector of all monomials of degree upto $i$  in $x$ (with the first entry equal to the 0-degree monomial $1$). For example,
$\Poly{x}{2}=\tran{\begin{pmatrix} 1 & x_1 & \ldots & x_n & x_1^2 & x_1x_2 & \ldots x_n^2 \end{pmatrix}}$. Let $\Vc=\tran{\begin{pmatrix} z & \Rep{V} & \Imp{V} \end{pmatrix}}$. Let $m$ be the size of $\Poly{\Vc}{4}$. We define a moment vector $y$ of the same size as $\Poly{\Vc}{4}$ and the linear operator on the space of all degree-$4$ polynomials in $\br{\Vc}$. 
\[\mathcal{L}_y\br{\sum_{i=1}^m c_i\mathrm{Poly}^4_i\br{\Vc}}=\sum_{i=1}^m c_iy_i\]
We also define the localizing matrix for $i=1,2$: $X_i = \Poly{\Vc}{i}\tranb{\Poly{\Vc}{i}}$. The moment relaxation of \eqref{eq:JSing} is given by the following feasibility problem:
\begin{subequations}
\begin{align}
&\text{Find } y \text{ such that }  \\
& \Ly{\JF\br{V}z{\begin{pmatrix} 1 & \tran{\Vc} \end{pmatrix}}}=0 \\
& \Ly{\br{|V_i-V_j|^2-{\br{\min\br{\fb_{ij},\gamma}}}^2}X_1}\succeq 0\\
& \Ly{\br{|V_i|^2-\Vset_i^2}X_1}= 0, i \in \G \\
& \Ly{\br{|V_0-1|^2}X_1}= 0,\Ly{\br{\tran{z}z-1}X_1}=0 \\
& \Ly{X_2}\succeq 0,y_1=1
\end{align}\label{eq:MomentRelax}
\end{subequations}
Each constraint here is either a linear equality constraint or a linear matrix inequality in $y$. If \eqref{eq:JSing} is feasible, then \eqref{eq:MomentRelax} is feasible as well: Let $\br{z,V}$ be feasible for \eqref{eq:JSing} and define $\Vc=\tran{\begin{pmatrix} z & \Rep{V} & \Imp{V} \end{pmatrix}}$. Then, $y=\Poly{\Vc}{4}$ is feasible for \eqref{eq:MomentRelax} by construction. Thus, if \eqref{eq:MomentRelax} is infeasible, then so is \eqref{eq:JSing}.

The moment relaxation of problem \eqref{eq:Fineq} can be formulated similarly. Suppose that $\Scal=\{s:\norm{s-s^0}\leq \delta\}$ (ellipsoids and polytopes in $s$ can be handled similarly). We define $\Vc=\tran{\begin{pmatrix} s & \Rep{V} & \Imp{V} \end{pmatrix}}$ and $y,\Poly{\Vc}{4},X_1,X_2$ similarly as before. We can then construct the moment relaxation of \eqref{eq:Fineq}:
\begin{subequations}
\begin{align}
&\text{Find } y \text{ such that }  \\
& \Ly{\Fc_i\br{V}X_1}= 0, \Ly{\br{F_j\br{V}-s_j}X_1}= 0,j=1,\ldots,\nf \\
& \Ly{\Fc_j\br{V}X_1}= 0, j=1,\ldots,|\Fc|,j\neq i\\
& \Ly{\Hc_j\br{V}X_1}\succeq 0,j=1,\ldots,|\Hc|\\
& \Ly{X_2}\succeq 0,y_1=1,\Ly{\br{\delta^2-\norm{s-s^0}^2}X_1}\succeq 0
\end{align}\label{eq:MomentRelax2}
\end{subequations}
Infeasibility of \eqref{eq:MomentRelax2} implies infeasibility of \eqref{eq:Fineq}, using a similar argument as before. However, if \eqref{eq:MomentRelax2} is feasible, one can extract $\br{V,s}$ from $y$. If this is feasible for \eqref{eq:Fineq}, then we know that $\Scal$ is not a \DSF, since it contains a point that is not strictly feasible.
\end{document}